\documentclass[a4paper]{article}
\usepackage{jheppub}
\usepackage[utf8]{inputenc}
\usepackage{amsmath,amsthm,amssymb,mathtools,booktabs,framed,eucal}

\newcommand{\mc}{\mathcal} \newcommand{\mf}{\mathfrak} \newcommand{\mb}{\mathbb} \newcommand{\on}{\operatorname}
 \newcommand{\slot}{\;\cdot\;} \newcommand{\ber}{\on{Ber}}
\usepackage{thmtools}
\usepackage[colorlinks = true, linkcolor = blue, citecolor = blue, anchorcolor = blue]{hyperref}

\declaretheorem[name=Lemma, numbered=no]{lem}

\declaretheorem[name=Remark, style=definition, numbered=no]{rem}
\title{Super AKSZ construction, integral forms, and the 2-dimensional $\mathcal N=(1,1)$ sigma model}
\subheader{\hfill\textrm{Imperial/TP/22/FV/2}}
\abstract{We discuss a natural extension of the AKSZ construction to the case where the source is given by a supermanifold with a chosen integral form. We then focus on the special case with the target given by a Courant algebroid. In the simplest case this leads to the BV version of the super Chern--Simons theory, as developed by Grassi--Maccaferri and Cremonini--Grassi. In the case of exact Courant algebroids we derive the 2-dimensional $\mathcal N=(1,1)$ sigma model on the boundary, together with the Wess--Zumino term, paralleling the approach of Ševera in the bosonic case.}
\author[a]{Ondřej Hulík,}
\author[b]{Josef Svoboda,}
\author[c]{Fridrich Valach}
\affiliation[a]{Theoretische Natuurkunde, Vrije Universiteit Brussel\\ Pleinlaan 2, B-1050 Brussels, Belgium}
\affiliation[b]{Department of Mathematics, University of Miami\\Coral Gables, FL 33146, USA}
\affiliation[c]{Department of Physics, Imperial College London\\Prince Consort Road, London, SW7 2AZ, UK}
\emailAdd{ondra.hulik@gmail.com}
\emailAdd{josefsvobod@gmail.com}
\emailAdd{f.valach@imperial.ac.uk}

\makeatletter
\gdef\@fpheader{}
\makeatother

\begin{document}
\maketitle
\section{Introduction}
The AKSZ construction \cite{AKSZ} provides an elegant geometric description of a large class of topological field theories in the Batalin--Vilkovisky (BV) formalism. The key point is that these theories can be cast in a form that generalises the standard nonlinear sigma models, with the fields now being maps between differential graded manifolds instead of ordinary manifolds. As particular cases one obtains the topological A and B-models and the Chern--Simons theory together with its generalisations, the so-called Courant sigma models \cite{Ikeda,Roytenberg2}.

More generally, the AKSZ construction for manifolds with boundaries, with suitable boundary conditions imposed, can be used to produce interesting non-topological theories. For instance, this route was taken in \cite{Severa3} to provide a more conceptual interpretation of the Poisson--Lie T-duality \cite{KS}.

Applying the AKSZ construction in the supersymmetric case is more subtle. One difficulty stems from the fact that general $\mb N\times\mb Z_2$-graded manifolds do not carry any non-degenerate measure \cite{Salnikov}. To consider this in more detail, note that the typical source space in the AKSZ construction is given by the shifted tangent bundle $T[1]M$ of an ordinary (oriented) manifold $M$. Functions on this shifted bundle can be naturally identified with differential forms on $M$ and hence can be integrated --- this provides a measure on $T[1]M$.

The identification of functions on $T[1]M$ with differential forms on $M$ is valid also when $M$ is a supermanifold. However, differential forms on supermanifolds can no longer be naturally integrated (in other words, the ``natural'' integral is infinite). One way out is to work with integral forms \cite{BL}, which form a module over the ring of differential forms. This is the approach taken in the present paper.

More precisely, we fix an integral form, considering it as a part of the initial data for defining a theory. This provides a well-defined integral of differential forms (simply by multiplying differential forms with the chosen integral form and then integrating). Even though this construction does not lead to a proper measure, we will argue that the AKSZ procedure can still be carried through, with only a minor modification. Our main point is then the construction of the 2-dimensional $\mc N=(1,1)$ sigma model, arising on a boundary of a $3|2$-dimensional topological Courant sigma model, paralleling the work of Ševera \cite{Severa3} in the bosonic case.

One advantage of this approach is the natural appearance of the full version of the model, including the Wess--Zumino term \eqref{eq:final} depending on the choice of the integral form.
The formalism and results developed in this work provide also a natural framework for the study of the Poisson--Lie T-duality (following \cite{Severa3}) for $\mc N=(1,1)$ models. We leave this to a future work.

The paper is organised as follows. We start by reviewing briefly the theory of NQ and differential graded (dg) manifolds and supermanifolds. We proceed to discuss the super AKSZ construction, first in the setup without boundary and then in the presence of boundary. We illustrate the construction in the case of super Chern--Simons theory, following \cite{GM} and \cite{CG}. Afterwards, we discuss in detail the geometric structures present on the superstring worldsheet. We finish by reconstructing the $\mc N=(1,1)$ model. We include two appendices, one with a brief introduction to the theory of integral forms, and the second one with a coordinate-free description of the canonical integral form on the worldsheet.
All (super)manifolds appearing in this text are taken to be oriented. 

\subsection*{Acknowledgement}
The authors would like to thank Alberto S.\ Cattaneo for a helpful discussion.
O.\,H.\ was supported by the FWO-Vlaanderen through the project G006119N and by the Vrije Universiteit Brussel through the Strategic Research Program ``High-Energy Physics''. F.\,V.\ was supported by the Early Postdoc Mobility grant P2GEP2\underline{\phantom{k}}188247 and the Postdoc Mobility grant P500PT\underline{\phantom{k}}203123 of the Swiss National Science Foundation.

\section{Prerequisites}
    We assume some basic familiarity of the reader with supermanifolds. (For a pedagogical exposition on this topic we recommend \cite{Witten}.) For convenience we include a short review of integral forms in Appendix \ref{ap:if}. Understanding a supermanifold as a space whose algebra of functions is $\mb Z_2$-graded, we can make a small modification of the concept and define NQ and dg manifolds, as follows.
  \subsection{NQ and dg (super)manifolds}
    An \emph{NQ manifold} is a space $M$ whose algebra of functions is $\mb N$-graded\footnote{We take $\mb N:=\mb Z_{\ge 0}$.}, equipped with a vector field $Q$ of degree $1$ (sometimes called the \emph{differential}), which satisfies $Q^2=0$. Every NQ manifold carries a canonical \emph{Euler vector field} $E$, which acts on functions of a homogeneous degree by $Ef=(\deg f)f$. The degree zero part of an NQ manifold is called its \emph{base} and is denoted by $M_0$.
    
    An easy example of an NQ manifold is the shifted tangent bundle $T[1]N$ of an ordinary manifold $N$, where the notation means that the degree of the fiber coordinates is shifted by 1. Functions on $T[1]N$ can be identified with differential forms on $N$. The vector field $Q$ is given by the de Rham differential.
    
    On any NQ manifold, $Q$ vanishes on the base of $M$. Consequently, it induces a complex on the tangent space at $M_0$, called the \emph{tangent complex}. An NQ manifold is called \emph{acyclic} if the cohomology of this complex vanishes at every point.
    
    An \emph{NQ (pre)symplectic manifold} is an NQ manifold with a (pre)symplectic form of some definite degree $n$ (i.e.\ $\mc L_E\omega=n\omega$) such that $\mc L_Q\omega=0$.
    
    More generally, in the case where the $\mb N$-grading is replaced by a $\mb Z$-grading, we talk instead about \emph{dg manifolds} and \emph{dg (pre)symplectic manifolds}, respectively.
    
    On any dg pre-symplectic manifold with $n\neq -1$ we have that the vector field $Q$ is Hamiltonian, for some function $H$ of degree $n+1$, i.e.\ $i_Q\omega=dH$. (We can just take $H=\tfrac1{n+1}i_Ei_Q\omega$.) The condition $Q^2=0$ translates to the \emph{classical master equation} $\{H,H\}=0$.
    
    In the singular case $n=-1$, a dg symplectic manifold with a choice of Hamiltonian for $Q$ is called a \emph{classical BV manifold}.
    
    In this paper, we will be interested in the merger of the two worlds, namely we will work also with \emph{NQ/dg (presymplectic) supermanifolds}, which carry an extra $\mb Z_2$-grading, independent of the $\mb N/\mb Z$-grading --- these spaces will be denoted by calligraphic letters. The commutative properties of objects (functions, forms, etc.) on a dg supermanifold are governed by the \emph{total parity}, that is by the sum of the $\mb N/\mb Z$ and $\mb Z_2$-parities. For instance, we require that $Q$ is $\mb Z_2$-even, so that its overall parity is odd. Similarly, we require $\omega$ to be $\mb Z_2$-even.
  \subsection{Courant algebroids}\label{subsec:CA}
    NQ symplectic manifolds with $n=2$ are called \emph{Courant algebroids} \cite{LWX,Severa, Roytenberg}. In this case, we can (locally) choose coordinates $x^i$, $e^\alpha$, $p_i$ of degrees $0$, $1$, and $2$, respectively, such that
    \[\omega=dp_idx^i+\tfrac12h_{\alpha\beta}de^\alpha de^\beta,\]
    with $h_{\alpha\beta}$ constant. Furthermore, the Hamiltonian has the form
    \[H=\rho^i_\alpha(x)e^\alpha p_i-\tfrac16 c_{\alpha\beta\gamma}(x)e^\alpha e^\beta e^\gamma,\]
    for some $\rho$ and $c$, constrained by the classical master equation.
    
    There are two particularly important classes of Courant algebroids. First one comes from \emph{quadratic Lie algebras} $\mf g$, i.e.\ Lie algebras with an invariant non-degenerate symmetric bilinear form. This induces a Courant algebroid structure on $\mf g[1]$ given as follows. We can interpret any basis $e^\alpha$ of $\mf g^*$ as linear coordinates on $\mf g[1]$ of degree 1 --- the symplectic form and the Hamiltonian then take the form
    \begin{equation}\label{eq:qla}
    \omega=de^\alpha d e_\alpha,\qquad H=-\tfrac16 c_{\alpha\beta\gamma}e^\alpha e^\beta e^\gamma,
    \end{equation}
    where $c$ are the structure constants of the Lie algebra and we have used the bilinear form to lower the indices.
    
    The second interesting class is given by acyclic (also known as \emph{exact}) Courant algebroids. The famous result of \v Severa \cite{Severa,Severa2} says that these always take the form
    \begin{equation}\label{eq:ex}
      M\cong T^*[2]T[1]M_0,
    \end{equation}
    with the standard symplectic form on the cotangent bundle, and with the Hamiltonian given by
    \[H=d-\eta,\]
    for some $\eta\in\Omega_{cl}^3(M_0)$. Here we understand the de Rham differential $d$ as a vector field on $T[1]M_0$ and hence as a linear function (of degree 3) on $T^*[2]T[1]M$. Similarly, $\eta$ is understood as a function on $M$ pulled back from $T[1]M_0$ along the projection map $T^*[2]T[1]M_0\to T[1]M_0$. However, the identification \eqref{eq:ex} is not unique --- it is easy to see that different choices of this identification lead to $\eta$'s differing by exact 3-forms. This leads to the classification of exact Courant algebroids over a given $M_0$, by $H^3(M_0,\mathbb{R})$. 
    
    Explicitly, if we pick coordinates $x^i$ on $M_0$, we automatically get a set of coordinates $x^i$, $\xi^i$, $\pi_i$, and $p_i$ on $T^*[2]T[1]M_0$ of degrees $0$, $1$, $1$, and $2$, respectively. We then have
    \[\omega=dp_idx^i+d\pi_id\xi^i,\qquad H=p_i\xi^i-\tfrac16\eta_{ijk}(x)\xi^i\xi^j\xi^k.\]
    
    Finally, a \emph{generalised metric} on an (arbitrary) Courant algebroid is a symplectic involution which preserves the base, i.e.\ a diffeomorphism $R$ with $R^*\omega=\omega$, $R^2=\on{id}$, and $R|_{M_0}=\on{id}_{M_0}$.
    One can always (locally) find adapted coordinates $x^i$, $e^a$, $e^{\dot a}$, $p_i$ of degrees $0$, $1$, $1$, $2$, such that
    \[\omega=dp_idx^i+\tfrac12h_{ab}de^ade^b+\tfrac12h_{\dot a\dot b}de^{\dot a}de^{\dot b},\qquad R^* x^i=x^i,\quad R^*e^a=e^a,\quad R^* e^{\dot a}=-e^{\dot a},\quad R^*p_i=p_i,\]
    with $h_{ab}$ and $h_{\dot a\dot b}$ constant (see \cite{Valach}). For simplicity, we shall also demand throughout the text that there is the same number of $e^a$'s and $e^{\dot a}$'s, and also that $h_{ab}$ is a positive definite matrix.\footnote{Neither of those assumptions (which are standard in the literature, as they correspond to Riemannian setups) are really necessary, but they will simplify the exposition.}
    
    Finally, for an exact Courant algebroid equipped with a generalised metric, we can find a unique identification \eqref{eq:ex} for which we have
    \begin{equation}\label{eq:ECAmetric}
      R^*x^i=x^i,\qquad R^*\xi^i=(g^{-1})^{ij}\pi_j,\qquad R^*\pi_i=g_{ij}\xi^j,\qquad R^*p_i=p_i,
    \end{equation}
    for some Riemannian metric $g$ on $M_0$ (see e.g.\ \cite{SV}). The data of a generalised metric on an exact Courant algebroid thus translates into a pair $(g,\eta)$, where $\eta$ is the concrete closed 3-form associated to the identification \eqref{eq:ex}.
    
\section{Super AKSZ construction}
\subsection{The case without boundary} \label{subsec:without}  
  In order to define a super version of the AKSZ model \cite{AKSZ}, we will require the following choice of data:
  \begin{itemize}
    \item a closed supermanifold $\mc Y$ with a closed integral form $\mu$ of codimension $m$,
    \item a (closed) NQ symplectic manifold $M$ with $\deg \omega=n$ (we will assume $n\neq0$).
  \end{itemize}
  First, from $\mc Y$ we construct an NQ supermanifold $T[1]\mc Y$. The integral form induces a ``pseudo-measure'' on $T[1]\mc Y$, given by \[\int_{T[1]\mc Y} f:=\int_{i.f.} f\mu,\]
  which satisfies $\int_{T[1]\mc Y} Q f=0$. This is not an honest measure, since there exist functions $f$ which satisfy $\int_{T[1]\mc Y} fg=0$ for any other function $g$ --- we will call such functions $f$ \emph{degenerate}. Nevertheless, we will see that the pseudomeasure will be sufficient for the present construction.
  
    Consider now the space $\on{Maps}(T[1]\mc Y,M)$, consisting of all maps (not only the degree-preserving ones) --- this is an infinite-dimensional dg supermanifold, with the differential induced by the differentials on $T[1]\mc Y$ and $M$. The symplectic form $\omega$ on $M$, together with the ``pseudo-measure'' on $T[1]\mc Y$ induce a presymplectic form on $\on{Maps}(T[1]\mc Y,M)$ of degree $n-m$, given simply by
    \[T_\varphi\on{Maps}(T[1]\mc Y,M)\cong \Gamma(\varphi^*TM)\ni V,W\quad \mapsto\quad \int_{T[1]\mc Y}\omega(V,W).\]
    The differential on $\on{Maps}(T[1]\mc Y,M)$ is Hamiltonian, given by the AKSZ action
    \begin{equation}\label{eq:akszion}
    S(\varphi)=\int_{T[1]\mc Y} i_Q\varphi^*\alpha-\varphi^*H,
    \end{equation}
    where $\alpha$ is any potential for $\omega$, i.e.\ $d\alpha=\omega$ (this always exists, for instance we can take $\alpha=\tfrac1n i_E\omega$).
    
  Finally, modding out by the null leaves of the presymplectic form, we obtain an honest (infinite-dimensional) dg symplectic supermanifold, with the Hamiltonian given by the same formula \eqref{eq:akszion}.
    In the particular case where $n=m-1$, we thus get a classical BV manifold.
  
\subsection{The case with boundary}
  Let us now modify the setup from above by adding a boundary. Focusing on the case with $n=m-1$, we now require the following data:
  \begin{itemize}
    \item a supermanifold $\mc Y$ with an integral form $\mu$ of codimension $n+1$, with a boundary $\partial\mc Y=\mc S$ (with $\dim\mc Y-\dim\mc S=1|0$),
    \item an NQ symplectic manifold $M$, with $\deg\omega=n$ (we will again assume $n\neq0$).
  \end{itemize}
  We can still produce a space of maps $\on{Maps}(\mc Y,M)$, but the resulting differential will not preserve the presymplectic form due to the presence of the boundary.
  One way to remedy this is to impose a boundary condition, as follows.
  
  First, note that $\mu|_\mc S$ gives rise to a ``pseudo-measure'' on the NQ manifold $\mc S$.
  By the arguments above, the space $\on{Maps}(T[1]\mc S,M)$ has a presymplectic form of degree 0.
  Let now $\mc L\subset \on{Maps}(T[1]\mc S,M)$ be a dg isotropic submanifold. Then
  \[\{\varphi\in\on{Maps}(T[1]\mc Y,M)\text{ s.t. } \varphi|_{T[1]\mc S}\in\mc L\},\]
  the space of maps with the boundary condition $\mc L$, is dg presymplectic.
  
  If the boundary condition is preserved by the Euler vector field on $M$ (and $n\neq 0$), we can write the Hamiltonian as
  \begin{equation}\label{eq:akszbdry}
  S(\varphi)=\int_{T[1]\mc Y} i_Q\varphi^*\alpha-\varphi^*H,
  \end{equation}
  with the specific choice $\alpha:=\tfrac1n i_E\omega$ (see \cite{PSV}). 
  
  Modding out by the null leaves of the presymplectic form, we again get a classical BV manifold.

\section{Example without boundary: Super Chern--Simons theory}
  In this section we apply the previously introduced super AKSZ construction with the target given by a simple Courant algebroid --- namely a shifted quadratic Lie algebra $\mf g[1]$ --- leading to the supersymmetric Chern--Simons theory as constructed by Grassi--Maccaferri \cite{GM}. However, since in the present framework we consider maps of arbitrary degree (i.e.\ not necessarily degree-preserving maps), we obtain automatically the full BV description of this theory, described in \cite{CG}, analogously to the standard BV formulation of the ordinary Chern--Simons.

The supermanifold $\mc Y$ is now given by the flat 3-dimensional $\mc N=1$ superspace, i.e.\ $\mc Y=\mb R^{3|2}$. Local coordinates are taken to be $x^a,\theta^\alpha$, with $a\in\{1,2,3\}$ and $\alpha\in\{1,2\}$, while the target has coordinates $e^A$, parametrizing the Lie algebra generators.\footnote{In this section we temporarily change the form (Latin, Greek, capital) of the indices, in order to match better the notation in \cite{GM}.}

The maps $T[1]\mathcal{Y} \rightarrow \mathfrak{g}[1]$ are denoted by $\mathcal{A}$. Decomposing this into components, we get
  \begin{equation}\label{eq:fullA}
    \mathcal{A} = A_0 + A_1 + A_2 + A_3+\ldots
  \end{equation}
In more physics terms these are just regular forms on the supermanifold $\mc Y$ with values in the Lie algebra $\mf g$, with the subscript denoting the form degree. The degree-preserving maps correspond to $\mc A=A_1$, while in the general (not necessarily degree-preserving) case we work with the full multiform $\mc A$.
Note that although the expansion \eqref{eq:fullA} continues indefinitely, the only components appearing in the action will be the first four, as spelled out in the expression. Their physical meaning is as follows:
\begin{center}
\begin{tabular}{c|c|c|c}
    $A_0$ & $A_1$ & $A_2$ & $A_3$
    \\
    ghosts & fields & antifields & antighosts
    \end{tabular}
\end{center}

The suitable integral form \cite{GM}, derived from the requirements of supersymmetry and closure, is 
  \begin{equation*}
    \mu=(dx^a + \gamma^a_{\gamma \delta} \theta^{\gamma} d \theta^{\delta} )
    (dx^b + \gamma^b_{\epsilon \zeta} \theta^{\epsilon} d \theta^{\zeta} )
    \gamma_{ab}^{\alpha \beta} \partial_{d\theta^\alpha} \partial_{d\theta^\beta}
    \left(
    \delta ( d\theta^{1} )
    \delta ( d\theta^{2} )
    \right).
  \end{equation*}
Plugging \eqref{eq:qla} into the general action \eqref{eq:akszion} gives 
\begin{equation*}
    S(\mc A)=\tfrac12\int_{i.f.} \mu (\mathcal{A}_A d\mathcal{A}^A + \tfrac{1}{3} c_{ABC} \mathcal{A}^A  \mathcal{A}^B  \mathcal{A}^C).
\end{equation*}
This is a full BV super description the of supersymmetric Chern--Simons theory. Writing the action in components, and using $\langle\slot,\slot\rangle$ for the inner product on $\mf g$, we have
\begin{equation}\label{eq:scsfull}
S=    \int_{i.f.} \mu \left( 
    \tfrac{1}{2} \langle A_1, d A_1\rangle+\tfrac{1}{6}\langle A_1,[A_1,A_1]\rangle
    + 
    \langle A_2, dA_0 \rangle + \langle A_2, [A_1,A_0]\rangle+\tfrac12\langle A_3,[A_0,A_0]\rangle
    \right).
\end{equation}
 If we keep only the degree-preserving part, i.e.\ take $\mc A=A_1$, we are left with the first two terms, corresponding to the Grassi--Maccaferri action \cite{GM}.

Writing the first two terms in terms of components produces
\begin{equation}\label{eq:components}
    \tfrac12\int_{\ber}
    \left(
    \gamma^{\alpha \gamma}_{ac} \epsilon^{abc} \langle A_{\alpha}, F_{b \gamma}\rangle 
    +
    \gamma^{\beta \gamma}_{bc} \epsilon^{abc} \langle A_{a}, F_{\beta \gamma}\rangle
    -
    \tfrac{1}{6} \gamma_a^{\alpha \beta}\langle A_{\alpha}, [A_{\beta}, A^a]\rangle \right)[dx^1,dx^2,dx^3| d\theta^1,d\theta^2 ].
\end{equation}
Following \cite{GM}, we now impose the conventional constraint $F_{\alpha\beta}=0$ (which corresponds to a subset of the equations of motion), which restricts the field $A_1$ to the form
\begin{equation*}
            A_1 = 
        \left( a_a + \lambda \gamma_a \theta +\dots \right) dx^a +
        \left( a_a (\gamma^a \theta)_{\alpha} + \tfrac12\lambda_{\alpha} \theta^1\theta^2 \right) d\theta^{\alpha}.
\end{equation*}
The fields $a_a$ and $\lambda_\alpha$ correspond to the correct degrees of freedom of the $3$-dimensional $\mc N=1$ multiplet, namely to the ordinary gauge field and gaugino, respectively, and \eqref{eq:components} reduces to
  \[S=\int d^3x \left(\tfrac{1}{2} \langle a, d a\rangle+\tfrac{1}{6}\langle a,[a,a]\rangle+\tfrac12\epsilon^{\alpha\beta}\lambda_\alpha\lambda_\beta\right).
\]

  As is usual in the BV framework, one can read off the gauge transformations of the field $A_1$ directly from the action \eqref{eq:scsfull}, by looking at the terms involving $A_2$ \cite{CG}. This means that 
  \begin{equation} \label{eq:trans}
    \delta A_1 = dc + [A_1,c],
  \end{equation}
  where $c$ is a Lie algebra-valued function on $\mb R^{3|2}$.
  However, we only wish to keep the gauge transformations that do not change the conventional constraint. This means that $c$ has to have the form
  \begin{equation*}
         c = \alpha  + (\theta \gamma^a \theta) \partial_{x^a} \alpha  
  \end{equation*}
  where the gauge parameter $\alpha$ only depends on $x^a$. Inserting this back into \eqref{eq:trans}, we obtain the correct transformation of the gauge field $a$:
\begin{equation*}
     \delta a = d\alpha + [a,\alpha].
\end{equation*}
Thus we see that \cite{GM} extends to an AKSZ model, in a way compatible with the conventional constraint.

\section{Worldsheet}
    We will now consider a special case, relevant for superstring theory. We follow the exposition in \cite{Witten2}.\footnote{A minor difference is that in the present text we consider the real version, instead of complex one discussed in \cite{Witten2}.} We will take our \emph{worldsheet} $\mc S$ to be a (real) supermanifold of dimension $2|2$, equipped with two complementary integrable distributions $\mc R$ and $\bar {\mc R}$ of rank $1|1$, each of which in turn containing a maximally nonintegrable distribution of rank $0|1$, denoted $\mc D$ and $\bar{\mc D}$. Maximal nonintegrability means that if (locally) $D$ is an everywhere nonvanishing section of, say $\mc D$, then $D^2=\tfrac12[D,D]$ and $D$ are linearly independent at every point. We thus have
    \[T\mc S=\mc R\oplus \bar{\mc R},\qquad \mc D\subset \mc R,\quad \bar{\mc D}\subset \bar{\mc R}.\]
  \subsection{Explicit description}
    In particular, the space $\mc S$ can be locally seen as a product $\mc U\times\bar{\mc U}$, with $\dim \mc U=\dim \bar{\mc U}=1|1$, such that $\mc R$ and $\bar{\mc R}$ can be identified with tangent bundles of $\mc U$ and $\bar{\mc U}$, respectively. Focusing on the first factor, we can locally find coordinates $\sigma$ and $\vartheta$ on $\mc U$ such that $\mc D$ is spanned by $\partial_\vartheta+\vartheta\partial_\sigma$, with $\sigma$ even and $\vartheta$ odd. Such coordinates are called \emph{superconformal}.
    
    Similarly, we can find coordinates $\bar\sigma$ and $\bar\vartheta$ on $\bar{\mc U}$ such that $\bar{\mc D}$ is spanned by $\partial_{\bar\vartheta}+\bar\vartheta\partial_{\bar\sigma}$.
  \subsection{Involution}    
    There are also some other important structures induced by the distributions on the worldsheet. For instance, we have an operator on $T^*\mc S$ which acts on $\mc R$ and $\bar{\mc R}$ as $1$ and $-1$, respectively. In terms of superconformal coordinates
    \[d\sigma \mapsto d\sigma,\qquad d\vartheta\mapsto d\vartheta,\qquad d{\bar\sigma}\mapsto -d{\bar\sigma},\qquad d{\bar\vartheta}\mapsto -d{\bar\vartheta}.\]
    This lifts to an involution $\star$ on $T[1]\mc S$, by acting trivially on the degree 0 coordinates.\footnote{The resemblance of $\star$ and the Hodge start $*$ is not accidental.}
  \subsection{Integral form}
    The worldsheet also comes equipped with a natural integral form. In terms of superconformal coordinates, this is given by $\mu=\nu\bar\nu$, where 
    \begin{equation}\label{eq:if}
      \nu=(d\sigma-\vartheta d\vartheta)\delta'\!(d\vartheta)
    \end{equation}
    and similarly for $\bar\nu$. This expression is independent of the choice of superconformal coordinates. Let us show this in the case of an infinitesimal coordinate transformation. A more conceptual proof can be found in the Appendix \ref{app:form}.
    
    The most general infinitesimal coordinate transformations on $\mc U$ preserving the superconformality are given by
    \[\delta \sigma=g(\sigma)-f(\sigma)\vartheta,\qquad \delta\vartheta=f(\sigma)+g'(\sigma)\tfrac{\vartheta}2,\]
    where $f$ and $g$ are odd and even functions, respectively \cite{Witten2}.\footnote{Both these functions can depend on extra even and odd parameters (``moduli''). In particular $f$ can be odd despite the fact that it does not have any $\vartheta$ dependence.} This corresponds to the vector field $U_f+V_g$, where
    \[U_f=f(\sigma)(\partial_\vartheta-\vartheta\partial_\sigma),\qquad V_g=g(\sigma)\partial_\sigma+g'(\sigma)\tfrac{\vartheta}2\partial_\vartheta.\]
    We want to show that $\mc L_{U_f}\nu=\mc L_{V_g}\nu=0$.
    First, we calculate\footnote{We use the fact that $\mc L_X\delta^{(m)}(d\vartheta)=(d\mc L_X\vartheta)\delta^{(m+1)}(d\vartheta)$.}
    \[\mc L_{U_f}\nu=[\mc L_{U_f}(d\sigma-\vartheta d\vartheta)]\delta'\!(d\vartheta)+(d\sigma-\vartheta d\vartheta)\mc L_{U_f}\delta'\!(d\vartheta)=-2f'd\sigma\,\vartheta\, \delta'(d\vartheta)-f'd\sigma\,\vartheta\, d\vartheta\,\delta''\!(d\vartheta),\]
    which vanishes since $y\delta''(y)=-2\delta'(y)$ for even $y$ (by the integration by parts). Second,
    \[\mc L_{V_g}\nu=[\mc L_{V_g}(d\sigma-\vartheta d\vartheta)]\delta'\!(d\vartheta)+(d\sigma-\vartheta d\vartheta)\mc L_{V_g}\delta'\!(d\vartheta)=(d\sigma-\vartheta d\vartheta)g'\delta'(d\vartheta)+\tfrac12(d\sigma-\vartheta d\vartheta)a'd\vartheta\,\delta''(d\vartheta),\]
    which is zero for the same reason.
    
  \subsection{Pseudo-measure}\label{subsec:measure}
    For a fixed choice of superconformal coordinates we have the frame \[\partial_\sigma \qquad\partial_{\bar\sigma}\qquad \partial_\vartheta+\vartheta\partial_\sigma\qquad \partial_{\bar\vartheta}+\bar\vartheta\partial_{\bar\sigma},\] and the corresponding dual coframe
    \[\chi:=d\sigma-\vartheta d\vartheta\qquad \bar\chi:=d\bar\sigma-\bar\vartheta d\bar\vartheta\qquad \psi:=d\vartheta\qquad \bar\psi:=d\bar\vartheta.\]
    In particular, this defines coordinates on $T[1]\mc S$ --- for convenience we list these, together with their $\mb N$-degrees and total parity:
    \begin{center}
      \begin{tabular}{c|c|c|c|c|c|c|c|c}
       & $\sigma$& $\bar\sigma$& $\vartheta$& $\bar\vartheta$& $\chi$& $\bar\chi$& $\psi$& $\bar\psi$\\\hline
       $\mb N$-degree & 0 & 0 & 0 & 0 & 1 & 1 & 1 & 1\\\hline
       total parity & e & e & o & o & o & o & e & e
      \end{tabular}
    \end{center}
    We again define the pseudo-measure on $T[1]\mc S$ by
    \[\int_{T[1]\mc S}f:=\int_{i.f.}f\,\mu.\]
    Explicitly, if $f$ is a function on $T[1]\mc S$, supported in a domain in $\mc S$ covered by some pseudoconformal coordinates, then by expanding $f$ in $\chi$'s and $\psi$'s we have
    \begin{equation}\label{eq:Sintegral}
      \int_{T[1]\mc S} f=\int_{\on{Ber}} f^{\psi\bar\psi}(\sigma,\bar\sigma,\vartheta,\bar\vartheta)\,[d\sigma,d\bar\sigma|d\vartheta,d\bar\vartheta],
    \end{equation}
    where $f^{\psi\bar\psi}$ is the coefficient of $f$ next to $\psi\bar\psi$.
\section{The \texorpdfstring{$\mc N=(1,1)$}{N=(1,1)} supersymmetric sigma model}
  We will now show how to obtain the $\mc N=(1,1)$ supersymmetric sigma model on the boundary of a Courant sigma model in the present super AKSZ context, following the construction from \cite{Severa3} in the bosonic case.
  
  Suppose we have the following data:
  \begin{itemize}
    \item a supermanifold $\mc Y$ of dimension $3|2$, with a closed integral form $\mu$ of codimension 3 and with a boundary $\mc S$ equipped with a worldsheet structure (in the above sense); we will assume that $\mu$ restricted to the boundary coincides with the worldsheet integral form,
    \item an exact Courant algebroid $M$ with a \emph{generalised metric} $R$.
  \end{itemize}
  
    We will now construct the boundary condition $\mc L\subset \on{Maps}(T[1]\mc S,M)$, using the worldsheet structure on $\mc S$ and the generalised metric on $M$. First, $\mc L$ lies in the subspace where all coordinates of positive degree have been set to zero,\footnote{If we choose local coordinates on the space of maps so that we can locally identify it with a graded vector space, then the subspace $\mc L$ lies in the non-negative degrees. Such boundary conditions were called \emph{ghostless} in \cite{PSV}.} while we keep all coordinates of negative degrees. This means that in order to define it, we only need to specify the degree zero part of $\mc L$, i.e.\ \[\mc L_0:=\mc L\cap \on{Maps}_0(T[1]\mc S,M),\] where $\on{Maps}_0$ stands for the space of maps preserving the $\mb N$-degree. Note that any $\mc L$ of this form will be automatically preserved by the differential on $\on{Maps}(T[1]\mc S,M)$.
    
    We then set 
    \begin{equation}\label{eq:bccsm}
      \mc L_0=\{\varphi\in\on{Maps}_0(T[1]\mc S,M)\mid R\circ \varphi=\varphi\circ \star\}.
    \end{equation}
    
    More concretely, choosing the identification $M\cong T^*[2]T[1]M_0$ for which we have \eqref{eq:ECAmetric}, the boundary condition implies that on the boundary
    \[p_i\in \Omega^{\ge 2}(\mc S), \qquad \pi_i\in \Omega^{\ge 1}(\mc S), \qquad \xi^i\in \Omega^{\ge 1}(\mc S), \qquad x^i\in \Omega^{\ge 0}(\mc S),\]
    while the lowest-degree components, which we will denote by $\bar{p}_i$, $\bar{\pi}_i$, $\bar{\xi}^i$, $\bar{x}^i$ --- corresponding to maps preserving the $\mb N$-grading --- are constrained by \eqref{eq:bccsm}. The latter, in particular, keeps $\bar x^i$ arbitrary while it imposes\footnote{Strictly speaking, since $\star$ now acts on functions, we should correctly write $\star^*$. We hope the kind reader will forgive us for not doing that.} \[\bar\pi_i=\star g_{ij}(\bar x)\bar\xi^j.\]
    
    Since $\mc L$ is preserved by the Euler vector field on $M$, we get a BV space with the Hamiltonian given by \eqref{eq:akszbdry}, which becomes
    \[S=\int_{T[1]\mc Y}p_idx^i+\tfrac12 \pi_id\xi^i+\tfrac12\xi^id\pi_i-p_i\xi^i+\tfrac16 \eta_{ijk}(x)\xi^i\xi^j\xi^k,\]
    for $x^i$, $\xi^i$, $\pi_i$, $p_i$ differential forms on $\mc Y$ of arbitrary degree, constrained only by the boundary condition.
    Using $\int_{T[1]\mc Y}\, df=\int_{T[1]\mc S} f$, we get
    \[S=\int_{T[1]\mc Y}p_i(dx^i-\xi^i)+\pi_id\xi^i+\tfrac16 \eta_{ijk}(x)\xi^i\xi^j\xi^k+\tfrac12\int_{T[1]\mc S}\pi_i\xi^i.\]
    
    Following \cite{Severa3}, we now wish to integrate $p_i$ in the bulk $\mc Y$ out, i.e.\ see it as a Lagrange multiplier, imposing the constraint $\xi^i=dx^i$. There is however a small subtlety, arising from the fact that we do not have an honest measure. We are thus led to the following weaker constraint:
    \[\xi^i=dx^i+\Xi^i,\]
    where $\Xi^i$ is a degenerate function on $T[1]\mc Y$ (see Subsection \ref{subsec:without}). 
    Fortunately, it follows that both $d\Xi^i$ and $\Xi^i|_{T[1]\mc S}$ (and also $\star \Xi^i|_{T[1]\mc S}$) are degenerate as well, and hence $\Xi^i$ drops out from our action.
    
    Putting things together, using the boundary condition, and noting that all higher form components of $x^i$ drop out,\footnote{This is because any form of degree higher than 3 in the bulk and higher than 2 on the boundary are degenerate and they do not contribute to the integral.} we are left with the simple action
    \[S(x)=-\tfrac12\int_{T[1]\mc S}g_{ij}(\bar x)d\bar x^i\star\! d\bar x^j+\int_{T[1]\mc Y}\tfrac16\eta_{ijk}(\bar x)d\bar x^id\bar x^jd\bar x^k.\]
    
    Setting $y:=\bar x$ and using the frame introduced in Subsection \ref{subsec:measure}, we can write
    \begin{align*}
      dy^i&= \psi(\partial_\vartheta+\vartheta\partial_\sigma)y^i+\bar\psi (\partial_{\bar\vartheta}+\bar\vartheta\partial_{\bar\sigma})y^i+\chi\partial_\sigma y^i+\bar\chi\partial_{\bar\sigma}y^i,\\
      \star dy^i&= \psi(\partial_\vartheta+\vartheta\partial_\sigma)y^i-\bar\psi (\partial_{\bar\vartheta}+\bar\vartheta\partial_{\bar\sigma})y^i+\chi\partial_\sigma y^i-\bar\chi\partial_{\bar\sigma}y^i.
    \end{align*}
    Finally, applying \eqref{eq:Sintegral} we obtain
    \begin{equation}\label{eq:final}
    S(y)=\int_{\on{Ber}}g_{ij}(y)D y^i \bar D y^j[d\sigma,d\bar\sigma|d\vartheta,d\bar\vartheta]+\int_{i.f.} \mu(y^*\eta),
    \end{equation}
    where the first integral goes over the worldsheet $\mc S$, the second over the bulk $\mc Y$, and we introduced the usual supersymmetric derivatives $D:=\partial_\vartheta+\vartheta\partial_\sigma$ and $\bar D:=\partial_{\bar\vartheta}+\bar\vartheta\partial_{\bar\sigma}$. Recall that $y$ is now a map from $\mc S$ (or $\mc Y$) to $M_0$. This is the $\mc N=(1,1)$ supersymmetric string sigma model, with the Wess--Zumino term included.
    
    To get some insight into the nature of this term, note that if $\eta=dB$ for some $B$, we can rewrite the action as 
    \begin{equation}\label{eq:finalsimple}
    S(y)=\int_{\on{Ber}}(g_{ij}+B_{ij})D y^i \bar D y^j[d\sigma,d\bar\sigma|d\vartheta,d\bar\vartheta],
    \end{equation}
    which coincides with the standard action \cite{GHR}.
    
    \begin{rem}
    Formula \eqref{eq:finalsimple} shows that if $\eta$ is exact, then the action is independent of $\mu$, as long as we require that the restriction of $\mu$ to the boundary recovers the canonical worldsheet integral form. The dependence of \eqref{eq:final} on $\mu$ is a bit more subtle: Suppose $\mu$ and $\mu'$ are two closed integral forms which both restrict to the same worldsheet integral form, and call the respective actions \eqref{eq:final} by $S_\mu(y)$ and $S_{\mu'}(y)$. Then the difference $\Delta S(y):=S_\mu(y)-S_{\mu'}(y)$ is invariant under infinitesimal deformations of the map $y$, since when deforming $y$ in the direction of the flow of a vector field $V$ on $M_0$, we have
    \[\delta (\Delta S(y))=\int_{i.f.}(\mu-\mu') (y^*\mc L_V \eta)=\int_{i.f.}d[(\mu-\mu') (y^*i_V \eta)]=0,\]
    since $\mu$ and $\mu'$ coincide on the boundary. Thus the choice of $\mu$ does not affect the equations of motion and consequently does not matter as far as the classical theory is concerned. However, the choice of $\mu$ in \eqref{eq:final} becomes important in the quantum theory, affecting the path integral.
    \end{rem}
    
\appendix{   
\section{About integral forms}\label{ap:if}
  We follow the exposition from \cite{Witten}.
  Let $\mc Y$ be a supermanifold of dimension $p|q$ and $\ber (\mc Y)$ the Berezinian line bundle. Choosing (local) coordinates $\sigma^1,\dots,\sigma^p,\vartheta^1,\dots,\vartheta^q$, with $\sigma$ even and $\vartheta$ odd, produces \[[d\sigma^1,\dots,d\sigma^p|d\vartheta^1,\dots,d\vartheta^q]\in\Gamma(\ber(\mc Y)).\]
  In general, sections of $\ber(\mc Y)$ can be integrated:
  First, on $\mb R^{p|q}$ we define
  \begin{equation}\label{eq:appdef}
  \int_{\on{Ber}} f(\sigma,\vartheta)[d\sigma^1,\dots,d\sigma^p|d\vartheta^1,\dots,d\vartheta^q]:=\int f^{top}d\sigma^1\dots d\sigma^p,
  \end{equation}
  where $f^{top}$ is top part (in terms of number of $\vartheta$'s) of $f$ and on the RHS we use the ordinary integral. This gives a well-defined integral of sections of $\ber(\mc Y)$ which are supported in some coordinate chart, identified with a subset of $\mb R^{p|q}$. To get an integral of an arbitrary section of $\ber(\mc Y)$, we employ the standard procedure using the partition of unity and the linearity of the integral.
  
  Notice, however, that there is no well-defined integral of differential forms on $\mc Y$, due to the lack of top forms (since $d\vartheta$ are even). One instead has to introduce \emph{integral forms} \cite{BL}, as follows.
  
  First, note that differential forms can be seen as functions on $T[1]\mc Y$ (this is a space with local coordinates $\sigma^i,\vartheta^a,d\sigma^i,d\vartheta^a$, which are even, odd, odd, and even, respectively). We then define \emph{integral forms} as distributions on $T[1]\mc Y$, which are supported at the locus $d\vartheta^1=\dots=d\vartheta^q=0$. Any such object can be written using Dirac delta functions and their derivatives as a sum of terms of the form
  \[f_{i\dots j}(\sigma,\vartheta)d\sigma^i\dots d\sigma^j\delta^{(c_1)}(d\vartheta^1)\dots\delta^{(c_q)}(d\vartheta^q).\]
  Integral forms of \emph{top degree} are those of the form
  \[f(\sigma,\vartheta)d\sigma^1\dots d\sigma^p\delta(d\vartheta^1)\dots\delta(d\vartheta^q).\]
  All other integral forms are obtained by applying some number of $d/d\sigma^i$ or $d/d\vartheta^a$ operations to a top form. This number is then called the \emph{codimension} of the integral form. Note that multiplying an integral form of codimension $n$ with a differential form of degree $m$, we get an integral form of codimension $n-m$.
  
  To construct an integral of an integral form, we use the canonical section $s$ of $\ber(T[1]\mc Y)$, given locally by
  \[s=[d\sigma^1,\dots,d\sigma^p,d(d\vartheta^1),\dots,d(d\vartheta^q)|d\vartheta^1,\dots,d\vartheta^q,d(d\sigma^1),\dots,d(d\sigma^p)].\]
  We define the integral of an integral form $\mu$ using a Berezin integral on $T[1]\mc Y$,
  \[\int_{i.f.}\mu:=\int_{\on{Ber}}\mu \,s.\]
  This is nonzero only if $\mu$ is a top integral form. For instance, on $\mc Y=\mb R^{p|q}$ we get
  \[\int_{i.f.}f(\sigma,\vartheta)d\sigma^1\dots d\sigma^p\delta(d\vartheta^1)\dots\delta(d\vartheta^q)=\int f^{top}d\sigma^1\dots d\sigma^p,\]
  reproducing the integral from \eqref{eq:appdef}. In general, we have a bijective map from sections of the Berezinian to top-degree integral forms, \[\Gamma(\ber(\mc Y))\to \on{IF^{top}}(\mc Y),\qquad [d\sigma^1,\dots,d\sigma^p|d\vartheta^1,\dots,d\vartheta^q] \mapsto d\sigma^1\dots d\sigma^p\delta(d\vartheta^1)\dots\delta(d\vartheta^q).\]
  
\section{About the integral form}\label{app:form}
    Suppose $\mc U$ is any supermanifold of dimension $1|1$, together with a maximally non-integrable distribution $\mc D$ of rank $0|1$ (this corresponds to one of the ``halves'' of the worldsheet).
    We want to describe a coordinate-free interpretation of the integral form \eqref{eq:if}.
    
    Let us first show that there is a canonical isomorphism $\ber(\mc U)\cong \mc D^*$, following \cite{Witten2}:
    
    Due to the maximal nonintegrability, if $D\in\Gamma(\mc D)$ is nonvanishing, so is the image of $D^2$ in $\Gamma(T\mc U/\mc D)$. Furthemore, for any function $f$ we have $(fD)^2=f^2D^2+f(Df)D=f^2D^2\,(\on{mod}\, \mc D)$. Thus we obtain a line bundle isomorphism $\mc D\otimes \mc D\cong T\mc U/\mc D$. In particular this implies $\ber T\mc U\cong \ber(\mc D\otimes\mc D)\otimes\ber\mc D\cong (\mc D\otimes \mc D)\otimes \mc D^*\cong \mc D$, on account of $\mc D$ and $\mc D\otimes \mc D$ being an odd and even line bundle, respectively. Thus $\ber \mc U\equiv\ber T^*\mc U\cong \mc D^*$.
    
    Let us now define the map
    \[\lambda\colon \Omega^1(\mc U)\to \Gamma(\mc D^*)\cong \Gamma(\ber(\mc U))\to \on{IF^{top}}(\mc U),\]
    where the first map corresponds to the canonical projection from $T^*\mc U$ to $\mc D^*$.
    We then have:
    \begin{lem}
      The integral form \eqref{eq:if} is the unique integral form $\nu$ of codimension 1 such that for any $\alpha\in\Omega^1(\mc U)$ we have $\alpha\,\nu=-\lambda(\alpha)$.
    \end{lem}
    \begin{proof}
      Pick some local superconformal coordinates $\sigma,\vartheta$. Define $\tau\in\Gamma(\mc D^*)$ by $\langle \tau,\partial_\vartheta+\vartheta\partial_\sigma\rangle=1$. Then the projection $T^*\mc U\to\mc D^*$ takes the form $d\vartheta\mapsto \tau$, $d\sigma\mapsto \vartheta \tau$, while the map $\Gamma(\mc D^*)\to \on{IF^{top}}(\mc U)$ is given by $\tau\mapsto d\sigma\delta(d\vartheta)$. Thus $\lambda(a\,d\sigma+b\,d\vartheta)=(a\vartheta+b) d\sigma\delta(d\vartheta)$.
      On the other hand, the multiplication by a general integral form of codimension 1, $\nu=c\,\delta(d\vartheta)+e\,d\sigma\delta'(d\vartheta)$, gives
      \[(a\,d\sigma+b\,d\vartheta)(c\,\delta(d\vartheta)+e\,d\sigma\delta'(d\vartheta))=a\,d\sigma\, c\,\delta(d\vartheta)-b\,e\,d\sigma\delta(d\vartheta).\]
      Thus, the condition $\alpha\cdot\nu=-\lambda(\alpha)$ gives $c=\vartheta$, $e=1$, and so \[\nu=\vartheta\delta(d\vartheta)+d\sigma\delta'(d\vartheta)=(d\sigma-\vartheta d\vartheta)\delta'(d\vartheta).\qedhere\]
     \end{proof}
}

\end{document}